\documentclass{article}
\usepackage[utf8]{inputenc}
\usepackage[margin=0.8in]{geometry}
\usepackage{tabularx}
\usepackage{booktabs}
\usepackage{wrapfig}

\usepackage{ijcai19}
\usepackage{soul}
\usepackage{url}
\usepackage[hidelinks, draft]{hyperref}
\usepackage[utf8]{inputenc}
\usepackage[small]{caption}
\usepackage{graphicx}
\usepackage{amsmath}
\usepackage{booktabs}
\usepackage{algorithm}
\usepackage{algorithmic}
\usepackage{tikz}
\usepackage{amsthm,mathtools}
\usepackage{tikz}
\usepackage{dsfont}
\usepackage{wrapfig}
\usepackage{bm}

\newtheorem{theorem}{Theorem}
\newtheorem{prop}{Proposition}
\title{Reducing selfish routing inefficiencies using traffic lights}

\author{
	Charlotte Roman$^1$\footnote{Contact Author} \And
	Paolo Turrini$^2$ \\
	\affiliations
	$^1$Department of Mathematics, University of Warwick\\
	$^2$Department of Computer Science, University of Warwick\\
	\emails
	\{c.d.roman, p.turrini\}@warwick.ac.uk
}
\date{}
\begin{document}
\maketitle

\begin{abstract}
   Traffic congestion games abstract away from the costs of junctions in transport networks, yet, in urban environments, these often impact journey times significantly. In this paper we equip congestion games with traffic lights, modelled as junction-based waiting cycles, therefore enabling more realistic route planning strategies. Using the SUMO simulator, we show that our modelling choices coincide with realistic routing behaviours, in particular, that drivers' decisions about route choices are based on the proportion of red light time for their direction of travel. Drawing upon the experimental results, we show that the effects of the notorious Braess' paradox can be avoided in theory and significantly reduced in practice,  by allocating the appropriate traffic light cycles along a transport network.
\end{abstract}

\section{Introduction}
    Congestion games are the standard framework of algorithmic game theory employed to study the equilibria of traffic flows \cite{rough}. They are non-cooperative games of perfect information where self-interested actors choose sets of available resources, e.g., roads, and where the cost of each resource depends on its overall usage. 

    A well-known phenomenon occurring in these games is Braess' paradox \cite{Braess}, i.e., the existence of traffic networks that suffer from the increase of total cost when the cost of an available resource strictly decreases. 

    While Braess' paradox is an important mathematical result, its existence relies on rather constraining modelling assumptions, as congestion games abstract away from a number of important features of real-world road networks. Notably, their cost functions assume no clashes between antagonistic traffic flows which, in the real-world, are typically resolved by interdependent control mechanisms such as traffic lights. 

    Traffic lights are themselves an important object of research in Artificial Intelligence, as understanding their best configuration is paramount for the branch of AI concerned with optimising traffic \cite{Chouhan2018,Laszka2016,Lopez2018,Pol2016}. 
However, their effect on the traffic flow equilibria is yet to be understood.
    
\paragraph{Our contribution.} 
    Here, we formulate a congestion game model equipped with junction-based waiting cycles, encoding traffic lights, and prove that the equilibria of such games exist and are essentially unique. Using SUMO, the state of the art traffic simulation software, we show that the model is built on realistic assumptions in terms of the induced cost function and on realistic predictions of routing, as well. We then show that it is possible to use traffic lights to make a Braess-like network more efficient, noting that inaccurate configurations can lead to suboptimal network routing instead.

\paragraph{Related literature.}
    Congestion games have been a reference framework for transportation research since Wardrop \cite{Wardrop}, who established the conditions for a system equilibrium to exist such that all travellers have minimum and equal costs, leading to further important applications in the traffic domain \cite{Fisk,Pas,Sheffi,Yao2019,Zhao}.

    Following Braess' seminal work \cite{Braess}, the topic of routing paradoxes has been extensively explored \cite{Murchland,Pas,Zhao}. Pas and Principio \cite{Pas} classified demand constraints and linear cost functions that cause Braess' paradox on Braess-like networks. \citeauthor{Milchtaich} \cite{Milchtaich} singled out the topological conditions for an undirected network to be immune to Braess' paradox for a single population. In asymmetric games, conditions for network immunity were proposed by Chen et al. \cite{Chen}.  More generally, Epstein et al. \cite{Epstein} examined topologies in which every Nash equilibrium is socially optimal. Although we are restricting our analysis to perfect information games, there are connections with \cite{Meir2018}, who consider congestion games where players have subjective utilities and proved bounds on the impact bias players can have on the equilibria. 

    The study of real-world traffic networks heavily relies on computer-aided simulation. In this field,    
    Simulation of Urban Mobility (SUMO) \cite{Lopez2018} is the main open source traffic simulation software often used in transportation research \cite{Chouhan2018,Laszka2016,Lin2018,Lopez2018}. SUMO, which is our software of reference, is compatible with developing and testing intelligent traffic lights, and it is frequently used for the implementation and testing of reinforcement learning models \cite{Pol2016,Yang2017}.

\paragraph{Paper structure.}
    Section \ref{sec:preliminaries} introduces the basic results and definitions of congestion games needed later on. Section \ref{sec:traffic-lights} formulates the novel game model and shows that equilibria exist and are essentially unique. In Section \ref{sec:sumo}, we use SUMO simulations to illustrate that the game has realistic properties and assumptions. Finally, in Section \ref{sec:braess} we show how realistic traffic light models can be used to reduce selfish routing inefficiencies. 

\section{Mathematical preliminaries}\label{sec:preliminaries}

    Let $N=\{ 1,...,n\}$  be a nonempty finite set of agent populations, each travelling along a directed graph  $G=(V,E)$ between an origin $o \in V$ and destination $d \in V$. Agents belonging to the same population are assumed to have the same origin-destination (henceforth OD) pair. The demand for a population $i$, i.e., its size, is $d_{i}\geq 0$. Each population $i$ is assumed to control a nonempty finite resource set $E_i \subseteq E$, intuitively the edges they can potentially travel on, given their OD pair. More formally, we assume each $E_i$ to be made by those resources that are used in at least one of the population's potential OD routes, corresponding to strategy sets $S_{i} \subseteq 2^{E_{i}}$. We assume that each $e \in E$, i.e., each edge,  can be written as an ordered pair of nodes of $V$. Moreover, we assume that resource cost functions, i.e., functions of the form $c_e: \mathds{R}_{ \geq 0} \rightarrow \mathds{R}_{ \geq 0} \cup \{ \infty \}$ such that $e \in E$, are continuous, non-decreasing and non-negative. Finally, we define a \textit{(nonatomic) congestion game} as a tuple $\mathcal{M} = (N,(E_i),(S_i), (c_e), (d_i))$, with $i \in N$ and $e \in E$. 
    We stress that even if real-world traffic works with a discrete number of individuals, the analysis can often be significantly simplified, while keeping sufficient generality, using the non-atomic form, which assumes players controlling negligible amount of flow. 
    The outcome of all players of population $i$ choosing strategies leads to a strategy distribution $\bm{x}^{i}$ satisfying $\sum_{s \in S_{i}} x^{i}_s = d_{i}$ and $x^{i}_s \geq 0, \,  \forall s \in S_{i}$. 
	
    Call a strategy distribution or outcome $\bm{x} = (\bm{x}^{i})_{i \in N}$  \emph{feasible} if $\sum_{s \in S_{i} } x^{i}_s = d_{i}, \, \forall i \in N$. 
	Denote moreover the load on $e$ in an outcome $\bm{x}$ to be $f_e(\bm{x}) =\sum_{i \in N} \sum_{s \in S_i} x^i_s\mathbf{1}_s(e)$ where $\mathbf{1}$ is the indicator function. Moreover, in $\bm{x}$, let a player from population $i$ receive a cost function $C(s,\bm{x})=C_{i}(s,\bm{x})=\sum_{e \in s} c_e(f_e(\bm{x}))$ when selecting strategy $s\in S_{i}$. 

    A \textit{user equilibrium} (UE), also known as Wardrop equilibrium, is a strategy distribution $\bm{x}$ such that every player of every population chooses a strategy with minimum cost. More formally, a UE is a strategy distribution $\bm{x}$ such that the following inequality holds true: $\sum_{e \in s_i}c_e(f_e(\bm{x})) \leq \sum_{e \in s'_i} c_e(f_e(\bm{x}))$ for all $s_i, s'_i \in S_i$ such that $\bm{x}^i_{s_i}>0, \, \forall i \in N$.
	
    The \textit{social cost} $SC(\bm{x})$  of $\bm{x}$ is the total cost incurred by all players, i.e., \[ SC(\bm{x}) =\sum_{i \in N} C_i (\bm{x})d_i = \sum_{e \in E}f_e(\bm{x}) c_e(f_e(\bm{x}))\]
    
    Strategy distribution $\bm{x}$ is said to be the \textit{social optimum} (SO) if it solves the following minimisation problem: $\min_{\bm{x}} SC(\bm{x}) \mbox{ such that} \sum_{s_i \in S_{i}} x_{i}^{s_i} = d_{i},\forall i \in N, x_{i}^{s_i} \geq 0$. It is often the case that the SO solution is different to the UE solution, as players want to maximise their own individual utility. The \textit{Price of Anarchy} (PoA) is the ratio between the social cost of an SO outcome and the worst social cost of a UE. For UE $\bm{y}$,  \[PoA = \frac{\arg\min_{\bm{x}} SC(\bm{x})}{\arg \max_{\bm{y}}SC(\bm{y})}\]
    
	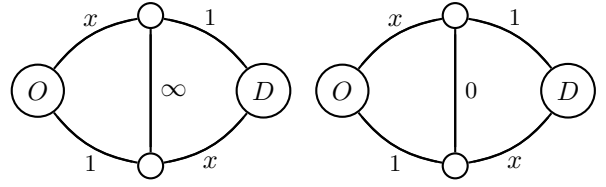
\begin{figure} [t]
		\begin{center}
			\begin{tikzpicture}[shorten >=2pt, thick]
			\node[circle,draw] (A) at (0,0) {$O$};
			\node[circle,draw] (B) at (1.5,1) {};
			\node[circle,draw] (C) at (1.5,-1) {};
			\node[circle,draw] (D) at (3,0) {$D$};
			\draw[-] (A) to[bend left=20] node[above] {$x$} (B);
			\draw[-] (B) to[bend right=20] node[above] {} (A);
			\draw[-] (B) to[bend left=20] node[above] {\small $1$} (D);
			\draw[-] (D) to[bend right=20] node[above] {} (B);
			\draw[-] (A) to[bend right=20] node[below] {\small $1$} (C);
			\draw[-] (C) to[bend left=20] node[below] {} (A);
			\draw[-] (C) to[bend right=20] node[below] {$x$} (D);
			\draw[-] (D) to[bend left=20] node[below] {} (C);
			\draw[-] (B) to node[right] {$\infty$} (C);
			\draw[-] (C) to node[right] {} (B);   
			\end{tikzpicture}\hspace{0.2cm}
			\begin{tikzpicture}[shorten >=2pt, thick]
			\node[circle,draw] (A) at (0,0) {$O$};
			\node[circle,draw] (B) at (1.5,1) {};
			\node[circle,draw] (C) at (1.5,-1) {};
			\node[circle,draw] (D) at (3,0) {$D$};
			\draw[-] (A) to[bend left=20] node[above] {$x$} (B);
			\draw[-] (B) to[bend right=20] node[above] {} (A);
			\draw[-] (B) to[bend left=20] node[above] {\small $1$} (D);
			\draw[-] (D) to[bend right=20] node[above] {} (B);
			\draw[-] (A) to[bend right=20] node[below] {\small $1$} (C);
			\draw[-] (C) to[bend left=20] node[below] {} (A);
			\draw[-] (C) to[bend right=20] node[below] {$x$} (D);
			\draw[-] (D) to[bend left=20] node[below] {} (C);
			\draw[-] (B) to node[right] {\small $0$} (C);
			\draw[-] (C) to node[right] {} (B);
			\end{tikzpicture}
			\caption{Braess' paradox on the Wheatstone network. When $d=1$, the social cost of travel is $\frac{3}{2}$ before and $2$ after reducing the costs of the middle edge.}
			\label{fig:Braess}
		\end{center}
	\end{figure}
    
    Braess' paradox arises when the cost of a resource is strictly decreased yet results in a strict increase in the social cost of the equilibria. This can be observed in the Wheatstone network in Figure \ref{fig:Braess}. A set of systems $(E, S_i)_{i \in N}$ admits \textit{Braess' paradox}  if there are two (nonatomic) congestion games $\mathcal{M}=(N,E,(S_i),(c_e), (d_i))$ and $\mathcal{M}'=(N,E,(S_i),(c'_e), (d'_i))$, with $i \in N$ and $e \in E$, for which $c'_e(t) \leq c_e(t)$, $\forall t \geq 0$ and $d'_i \leq d_i,$ $\forall i \in N$, and two UE $\bm{x}$ and $\bm{x}'$, such that $ SC(\bm{x}) < SC(\bm{x}')$. If no such $\mathcal{M}$ and $\mathcal{M}'$ exist, then we say that $\mathcal{M}$ is \textit{immune} to Braess' paradox. 

	A \textit{network} $G=(V,E)$ is a directed graph with a maximum of one edge between any ordered pair of nodes and no self-loops. For a directed edge $uv \in E$ where $u,v \in V$, we understand that $uv$ starts at $u$ and ends at $v$. The \textit{in-degree} of a node $v \in V$, $v_{in}$, is the number of edges that terminate there. 
	A \textit{path} is an ordered collection of edges such that adjacent pairs of edges share a node.	
	
    
    If a network is \textit{two-terminal} then there is a single origin and destination pair for players to travel between. 

	A two-terminal network is \textit{series-parallel} if it is either a single edge, or it is composed recursively by joining two series-parallel networks in series or in parallel. 
	The following result, used later on, outlines the conditions for a network to be immune to Braess' paradox.
	\begin{theorem}[\cite{Milchtaich}] \label{theorem:milch}
	A two-terminal network $G$ is immune to Braess' paradox if and only if $G$ is series parallel.
	\end{theorem}
	
\section{Congestion games with traffic lights}\label{sec:traffic-lights}

    In networks where traffic lights exist at junctions, the cost of waiting can be intuitively divided in two parts: the cost of using the preceding edge if there was no traffic light at its end and the additional cost that the traffic light imposes on the driver. This section deals with the problem of establishing the exact form of such division.
    

    Formally, we decompose the expected travel time of an edge $c_e$, into the expected travel time $\bar{c}_e$ that would occur from traveling across $e$, and the expected waiting time $w_e$ that occurs at its end. The cost function for any edge $e$ can therefore be written as $c_e = \bar{c}_e+w_e$. Note that an edge which does not end in a traffic light has $w_e =0$.
	
	A traffic light cycle is formed of repeated phases of green and red light when traffic is either serviced or not: each direction of travel has a cycle of $t^e_r$ red seconds and $t^e_g$ green seconds. Furthermore, between the phase changes there is a period of amber light to warn drivers of the change. We assume that the amber cycles have a constant time (in our simulation we set this to 3 seconds). The amber light at the end of a green cycle is included in $t^e_g$ and the amber cycle at the end of a red cycle is included in $t_r$. At the end of every edge there exists a probability $p_e = \frac{t^e_r}{t^e_r+t^e_g} \in [0,1 \rbrack$ that there will be a red traffic light when this node is reached. The waiting time at a traffic light is a function of $p_e$ as well as congestion $x$.

    In general, the cost of waiting at a traffic light is also dependent on congestion. Here, however, we make the key assumption that we can separate the effects of congestion at a traffic light from that of normal travel. Technically, we write the cost function of using an edge as $c_e(x,p_e) = \bar{c}_e(x) + w_e(x,p_e)$.  Furthermore, we make the following simplifying assumptions.
    We assume driving congestion functions $\bar{c}_e: \mathds{R}_{ \geq 0} \rightarrow \mathds{R}_{ \geq 0} \cup \{ \infty \}$ to be continuous, non-decreasing and non-negative, while traffic light waiting functions $w_e: \mathds{R}_{ \geq 0} \times [0,1] \rightarrow \mathds{R}_{ \geq 0} \cup \{ \infty \}$ to be continuous, non-negative and non-decreasing in $x$. Moreover, we assume that if $p_e=0$, then $w_e(x,0)=0$ and if $p_e \in (0,1]$ then $w_e(x,p_e)>0$. 

    As will be clear from the experiments in Section \ref{sec:sumo}, our assumptions are not only technically desirable but also justified empirically. In fact they isolate the natural cost functions for congestion games with control mechanisms as emerging from the microscopic traffic simulation.

    Given all the above, we define a \textit{traffic light congestion game} as a tuple $\mathcal{M} = (N,(E_{i}),(S_{i}), (p_e),(c_e), (d_{i}))$ where $e \in E,\, i \in N$. A \textit{traffic light user equilibrium} (TLUE) is a strategy distribution $\bm{x}$ such that all players choose a strategy they expect to be minimum cost: $\forall i \in N$ and $s,s' \in S_i$ such that $x^i_{s}>0$, we have $C_i(s,\bm{x}) \leq C_i(s',\bm{x})$. Note that a UE is a special case of a TLUE where the traffic lights are always green upon arrival. 

As standard with the perfect information treatment of congestion games, let us now assume everyone is aware of the probabilities $p_e$ --- this can for example be learned from past experience or through intelligent vehicles. Then a strategy distribution $\bm{x}$ is a TLUE if and only if $\forall s \in S_i$ such that $x^i_{s}> 0$:
	\[  C_i(\bm{x})= \min_{s \in S_i} C_i(s,\bm{x}) \]

    The following result follows.	
	\begin{prop} \label{prop:ue}
        Feasible strategy distribution $\bm{x}$ is a TLUE solution if, and only if, for any feasible $\bm{\tilde{x}}$: \[ C(\bm{x})(\bm{x}-\bm{\tilde{x}}) \leq 0 \]
    \end{prop}
    \begin{proof}
    	Suppose that $\bm{x}$ is a TLUE. Let population $i$ play strategy $s$ in $\bm{x}$: $\bm{x}^i_{s}=d_i$. Then any strategy $s' \in S_i$ that has a higher cost than $C_i(s,\bm{x})$ does not occur in a TLUE.
    	\[  C_i(s',\bm{x}_{\sim i}) > C_i(s,\bm{x}_{\sim i}) \Rightarrow \bm{x}^i_{s'} = 0 \]
    	Any feasible $\bm{\tilde{x}}$ can be written as a deviation from $\bm{x}$ by a set of population $M \subseteq N$: $\bm{\tilde{x}}=(s_1,s_2,...,s_m,\bm{x}_{\sim M})$.  
    	Hence, any other feasible flow has a route cost at least as high as $\bm{x}$. 
    	\[ \bm{\tilde{x}} C(\bm{x})  \geq  \bm{x} C(\bm{x}) \, \, \, \,  \forall \bm{\tilde{x}} \]
    	Now suppose that the converse is true, that there exists a feasible $\bm{\tilde{x}}$ with cost less than that of $\bm{x}$: \[ \exists \bm{x}' > 0, \, \,  C(\bm{\tilde{x}}) < C(\bm{x}).\] 
    	Then there exists a set of populations $M$ whose deviations create this distribution: $\bm{\tilde{x}} = (s_1,...,s_m,\bm{x}_{\sim M})$.
    	If those populations reroute their flow along these cheaper routes then it will reduce the total cost by $C(\bm{x})\bm{x} - C(\bm{\tilde{x}})\bm{\tilde{x}} > 0$. Thus, $ \bm{\tilde{x}} C(\bm{x}) < \bm{x} C(\bm{x}) $. So the inequality does not hold if $\bm{x}$ is not a TLUE. Hence, the statements are equivalent.
    \end{proof}
 
    From Proposition \ref{prop:ue} we can readily derive that for any TLUE solution $\bm{x}$ and any feasible $\bm{\tilde{x}}$:
	\[ C(\bm{x})(\tilde{\bm{x}}-\bm{x}) \geq 0 \]

    Now we can write the following derivation:
	\[ \sum_{i \in N} \sum_{s \in S_i} C_i(s,\bm{x})(\tilde{\bm{x}}^i_{s}-\bm{x}^i_{s}) \geq 0  \] 
	\[ \sum_{i \in N} \sum_{s \in S_i} \sum_{e \in s} c_e(f_e(\bm{x}))(\tilde{\bm{x}}^i_{s}-\bm{x}^i_{s}) \geq 0  \]
	\[ \sum_{i \in N} \sum_{s \in S_i} \sum_{e \in s} [\bar{c}_e(f_e(\bm{x}))+w_e(f_e(\bm{x}),p_e)](\tilde{\bm{x}}^i_{s}-\bm{x}^i_{s}) \geq 0  \]
	\[ \sum_{e \in E} \sum_{i \in N} \sum_{s \in S_i} [\bar{c}_e(f_e(\bm{x}))+w_e(f_e(\bm{x}),p_e)](\tilde{\bm{x}}^i_{s}-\bm{x}^i_{s}) \geq 0  \]
	\[ \sum_{e \in E} [\bar{c}_e(f_e(\bm{x}))+w_e(f_e(\bm{x}),p_e)](f_e(\tilde{\bm{x}})-f_e(\bm{x}^i_{s})) \geq 0 \]
	This is equivalent to
    \[ \min_{\bm{x}} \sum_{e \in E} \int_0^{f_e(\bm{x})}\bar{c}_e(\bm{z})+w_e(f_e(\bm{x}),p_e) dz.\] 
    This means that if cost functions are monotonic and differentiable, as they are in our case, a solution exists and is essentially unique, i.e., all TLUE solutions have the same social cost.

\section{Experimental evidence} \label{sec:sumo}
    \begin{figure*}[t]
        \centering
        \includegraphics[width=.46\textwidth, scale=0.7]{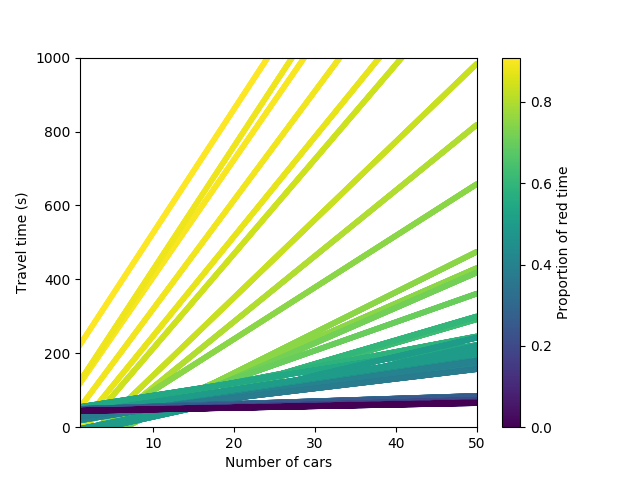}
        \includegraphics[width=.5\textwidth, scale=0.7]{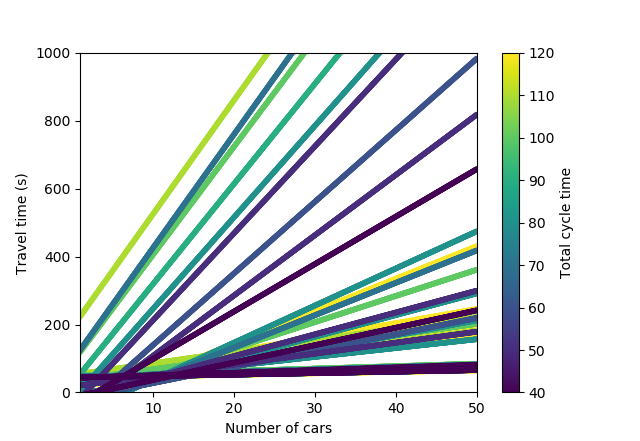}
          \caption{Simulated affine cost functions are the expected journey time found from regressions of simulated data points with a variety of values of $p \in [0,1]$ and $T\in [40,120]$. There is a clear correlation between the proportion of red light in a cycle $p$ and the expected travel times (left side). The length of the cycle time $T$ showed no correlation to journey times (right side). }
          \label{fig:costfunctions}
    \end{figure*}
    \begin{figure*}[th!]
        \centering
        \includegraphics[width=.45\textwidth]{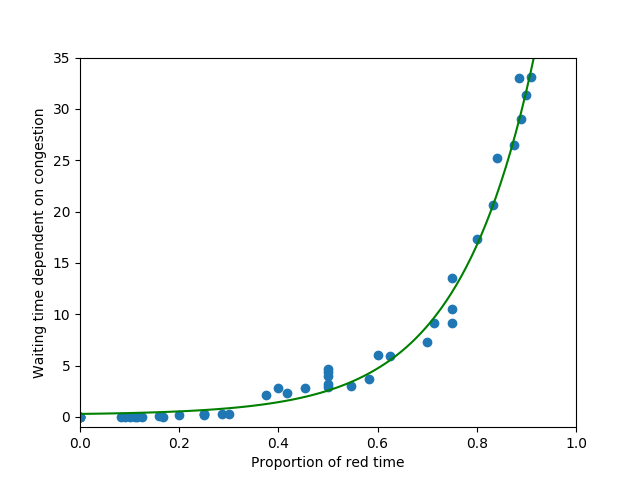}      \includegraphics[width=.45\textwidth]{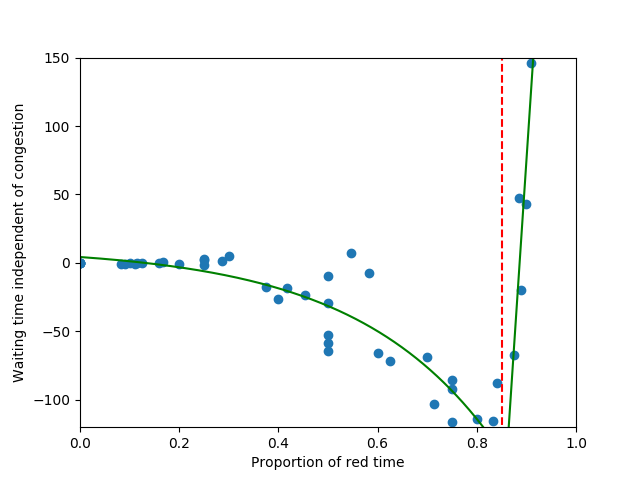}
        \caption{Waiting time dependent (left side) and independent (right side) of edge congestion as a function of $p$. The waiting time dependent of congestion is exponential as a function of the proportion of red time. The waiting time independent of congestion was not continuous on $p \in [0,1]$. For values of $p \in [0,0.85]$ the data fits an exponential function.}
        \label{fig:wp}
    \end{figure*}
    To verify the correctness of our assumptions over the cost functions in real-world cases, we have set up a SUMO simulation reproducing traffic journey over two horizontal edges with a junction between them, where the first edge was 500m long and the second edge was 50m long.
  

    The cost of the journey is taken as the total travel time and the congestion, $\bm{x}$, on an edge is the average number of cars that are present during the time a car is on the edge. Each simulation ran for 20000 timesteps with the same arrival distribution of cars. 

    Firstly, we had to find a compatible cost function for the journey that would occur if there was no traffic light at the junction. The affine cost function determined from the simulation was $\bar{c}(\bm{x}) = 0.4 \bm{x} + 43.9$. 
    
     The following tests were conducted to analyse the relationship between the waiting time $w_e$ and the traffic light cycle $\{t^e_r, t^e_g\}$. In each simulation there is fixed 3 second amber cycle to allow cars breaking time before a phase change. As discussed, if the amber light is after a green light then this is added to the green cycle time, otherwise, it is added to the red cycle time. For example, a cycle that is red for 17 seconds, then amber for 3 seconds, green for 17 seconds and finally amber for 3 seconds would be written as $t^e_r=20, t^e_g=20$. Any fixed traffic light cycle can be described through the total cycle time $T_e=t^e_r+t^e_g$ and the proportion of red time in a cycle $p_e=\frac{t^e_r}{t^e_r+t^e_g}$. We assume that $T_e$ belongs within some bounds to avoid scenarios where no car is allowed through at a traffic light, restricting the waiting time function to finite values. We compared the effects of changing these two variables on the cost functions and waiting times where the cost function of the journey can be written in the following form:
     \[ c(\bm{x}) = \bar{c}(\bm{x}) + w(\bm{x},T,p). \]
    
    Figure \ref{fig:costfunctions} plots the cost functions for a number of simulations. Each line represents the cost function found from linear regression of the data points. Lines are coloured using the $p$ values to indicate the correlation between waiting time and red cycle proportion. 

    From the simulations without a traffic light, we have $\bar{c}(x) =  0.4x + 44$. Once a traffic light was added, we collected the simulated data and used linear regression to find the waiting time functions. For example, for the cycle 20 seconds red and 20 seconds green we have $T=40$ $p=\frac{1}{2}$ and the simulation results were $c(x) = 4.4x - 9 =  \bar{c}(x)+ 4x - 53$. Hence, the waiting time coefficient dependent on congestion was $4$, and the time independent of congestion was $-53$.
    
    The waiting time coefficients were divided into either dependent or independent of congestion. The time dependent on congestion, shown in Figure \ref{fig:wp} (left side), follows an exponential distribution whereas the time independent of congestion, shown in Figure \ref{fig:wp} (right side), only follows an exponential distribution until a threshold value of red proportion and is then linear. This may be a result of acceleration constraining movement. After a threshold of 85\% red time, the waiting time increases significantly. This, observe, makes the cost function not continuous over $p \in [0,1]$, in general. However, it is safe to assume that in real traffic light systems none will have a cycle with more than 85\% red time. So, we have that for every $e \in E$ where there exists a traffic light, the proportion of red time is bounded $p_e \leq p^+ < 1$ for models with finite waiting time. Finally, we remark that there was no evidence to suggest a correlation between waiting time and $T$ (see Figure \ref{fig:costfunctions}).
    

    For our simulated network, the expected cost of the journey can be written as a function that is continuous and nondecreasing in $x$, as required.
    \[  c(\bm{x}) =  \bar{c}(\bm{x})+ (0.12e^{6.12p}+0.16)\bm{x}-10.51e^{3.05p}+14.51\]
    where $\bar{c}(\bm{x}) =  0.4\bm{x} + 44$ and $p<0.85$.


\section{Braess' paradox and traffic lights}\label{sec:braess}

    It is clear that changing traffic light cycles has an impact on the routing choices of drivers. In this section, we show that `biased' traffic lights can force optimal or suboptimal solutions.  
    To see this suppose that there exists a traffic light at a node $v$ if there are at least two edges entering $v$ and consider the directed Wheatstone graph with a traffic light as in Figure \ref{fig:trafficlightbraess}. We denote the presence of a traffic light by a square node. We can write the cost functions for the edges as follows: 
    \begin{align*}
    c_{e_1}(x) = & \bar{c}_{e_1}(x) & c_{e_4}(x) = &\bar{c}_{e_4}(x) \\
    c_{e_2}(x) = &\bar{c}_{e_2}(x)+w_{e_2}(x,p) &  c_{e_5}(x) = &\bar{c}_{e_5}(x)\\
    c_{e_3}(x) = &\bar{c}_{e_3}(x)+w_{e_3}(x,1-p)
    \end{align*}
  Assume moreover a TLUE $\bm{x}$ where $f_{e_1}(\bm{x}) = x$, $f_{e_2}(\bm{x}) = d-x$, $f_{e_3}(\bm{x}) = x-y$, $f_{e_4}(\bm{x}) = y$ and $f_{e_5}(\bm{x}) = d-y$. 
	\begin{figure} [t]
	    \centering
		\begin{tikzpicture}[shorten >=2pt, thick]
		\node[circle,draw] (A) at (-1.4,0) {$O$};
		\node[circle,draw] (B) at (1.5,1) {};
		\node[rectangle,draw] (C) at (1.5,-1) {};
		\node[circle,draw] (D) at (4.4,0) {$D$};
		\draw[->] (A) to[bend left=20] node[above] {$x$} (B);
		\draw[-] (B) to[bend right=20] node[above] {} (A);
		\draw[->] (B) to[bend left=20] node[above] {\small $1$} (D);
		\draw[-] (D) to[bend right=20] node[above] {} (B);
		\draw[->] (A) to[bend right=20] node[below] {\small $1+w(d-x,p)$} (C);
		\draw[-] (C) to[bend left=20] node[below] {} (A);
		\draw[->] (C) to[bend right=20] node[below] {$d-y$} (D);
		\draw[-] (D) to[bend left=20] node[below] {} (C);
		\draw[->] (B) to node[right] {$w(x-y,1-p)$} (C);
		\draw[-] (C) to node[right] {} (B);   
		\end{tikzpicture}
		\caption{The Wheatstone network with a single traffic light.}
		\label{fig:trafficlightbraess}
	\end{figure}
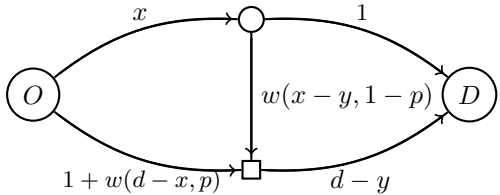

    If we increase $p$ it causes $w_{e_2}$ to decrease, therefore $d-x$ will increase. In addition, $w_{e_3}$ increases, hence, $x-y$ decreases. By increasing $p$ to a large value we can increase the costs of using the central edge to make its use undesirable to drivers. Suppose that $w(x,1)=\infty$. Then combining large $p$ with the congestion costs from the Braess example (as in Figure \ref{fig:trafficlightbraess}) limits players to the subgame where the social optimum solution is the TLUE. 
    
    Contrastingly, if we reduce $p$ then $w_{e_3}$ decreases which causes $x-y$ to increase. Furthermore, $w_{e_2}$ increases, hence, $d-x$ decreases. For very low values of $p$ combined with the Braess example cost functions, the TLUE here is for the players to all use the strategy $\{e_1,e_3,e_5\}$. 
    
    By adding only one traffic light into the Wheatstone network, we have shown that changing the proportions of red and green time of the traffic light can induce routing inefficiencies, or indeed, reduce the effects of Braess' paradox can be caused by traffic lights. Consequently, the traffic light cycles can be used either to force more socially optimal outcomes or to enable poor routing choices. This example shows the importance of traffic light cycles in urban areas as a tool to reduce congestion.
     
    We know from Theorem \ref{theorem:milch} that a two-terminal network is subject to Braess' paradox if it is not series-parallel. In order to reduce the effects of Braess' paradox, we must thus enforce the use of a subnetwork which is unaffected by Braess' paradox. In turn, in order to do that, we must find a set of ``undesirable'' edges to ``remove" so that the new network does not suffer from Braess' paradox. There are multiple ways of defining such a set, and we wish to determine the one of minimal size such that it is possible to deter the use of the edges through traffic lights.
    \begin{prop} \label{prop:2}
        For any two-terminal network that suffers from Braess' paradox, there exists a set of traffic light cycles $(p_e)_{e \in E}$ which can be combined with the network to remove Braess' paradox.
    \end{prop}
    \begin{proof}
        Find the maximal subgraph of $G$ which is series parallel, $\hat{G}=(V,E\backslash \hat{E})$. For any edge $\hat{e} \in \hat{E}$ there must exist a traffic light. (This is true since in order for the edge not to be added in series or in parallel, it must be formed from existing nodes. The in-degree of the node must already be one, so by adding a non-series parallel edge we must have a traffic light at the end node.) 
        For any such $\hat{e}$, set $p_{\hat{e}}=1$. For any edge $e \in E\backslash \hat{E}$, set $p_e= \frac{1}{v_{in}-v_{\hat{e}}}$ for $v_{in}-v_{\hat{e}}>1$ and $p_e=0$ otherwise, where $v_{\hat{e}}$ is the number of edges that end at $v$ that belong to $\hat{E}$. 
        By construction, the network $\hat{G}$ is series-parallel, hence, is immune to Braess' paradox.
    \end{proof} 
    Proposition \ref{prop:2} shows that any graph can be made immune to Braess' paradox by essentially closing roads by enforcing infinite waiting times on edges which tempt players to choose suboptimal routes.
     However, this result is highly theoretical. In Section \ref{sec:sumo}, we have shown that in real-world cases there are restrictions on the maximum and minimum possible values of $p$ due to the constraints of finite waiting times. Therefore, in real-life scenarios we may not be able to achieve immunity from Braess' paradox through traffic light bias, only reduce its impact. We now suggest that this improvement is restricted by the maximal in-degree of a node. To see this, consider the network in Figure \ref{fig:counterexample}.
        
	\begin{figure}[t]
	    \centering
		\begin{tikzpicture}[shorten >=2pt, thick]
		\node[circle,draw] (A) at (-1.4,0) {$O$};
		\node[circle,draw] (B1) at (-0.4,0.7) {};
		\node[circle,draw] (B2) at (0.5,0.9) {};
		\node[circle,draw] (B) at (1.5,1) {};
		\node[circle,draw] (B3) at (2.5,0.9) {};
		\node[rectangle,draw] (C) at (1.5,-1) {};
		\node[circle,draw] (D) at (4.4,0) {$D$};
		\draw[->] (A) to[bend left=20] node[above] {} (B1);
		\draw[-] (B1) to[bend right=20] node[above] {} (A);
		\draw[->] (B1) to[bend right=0] node[above] {} (B2);
		\draw[-] (B2) to[bend right=0] node[above] {} (B1);
		\draw[->] (B2) to[bend right=0] node[above] {} (B);
		\draw[-] (B) to[bend right=0] node[above] {} (B2);
		\draw[dashed,->] (B1) to[bend left=0] node[right] { } (C);
		\draw[dashed,->] (B2) to[bend left=0] node[right] { } (C);
		\draw[dashed,->] (B) to[bend left=0] node[right] { } (C);
		\draw[dashed,->] (B3) to[bend left=0] node[right] { } (C);
		\draw[->] (B) to[bend left=0] node[above] {} (B3);
		\draw[-] (B3) to[bend right=0] node[above] {} (B);
		\draw[->] (B3) to[bend left=10] node[above] {} (D);
		\draw[-] (D) to[bend right=10] node[above] {} (B3);
		\draw[->] (A) to[bend right=20] node[below] {} (C);
		\draw[-] (C) to[bend left=20] node[below] {} (A);
		\draw[->] (C) to[bend right=20] node[below] {} (D);
		\draw[-] (D) to[bend left=20] node[below] {} (C);
		\end{tikzpicture}
		\caption{In this network, the dashed lines need to be removed in order for immunity to Braess' paradox. The optimal allocation of $p_{\hat{e}}\leq p^+$ does not have immunity to Braess' paradox.}
		\label{fig:counterexample}
	\end{figure}
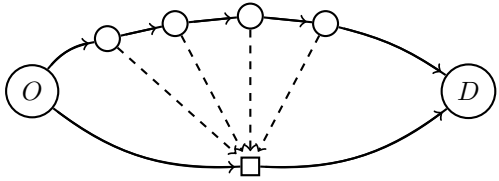
    
    In order to quantify the effects of implementing a traffic light, let us calculate the impact traffic light cycles have on altering the social cost of the equilibrium. Consider the cost functions from Figure \ref{fig:trafficlightbraess} with $d=1$. To find the socially optimum solution we must solve the following minimisation problem:
    \begin{align*}
    \min_{x,y} & [x^2 + (1-x) +(1-x)w(1-x,p) \\  & +(x-y)w(x-y,1-p)+y+(1-y)^2 ] 
    \end{align*}
    where $x \in [0,1]$ and $y \leq x$.
    Suppose the waiting functions are of the simplest exponential form: $w(x,p)=x(e^p-1)$. Note that the waiting time here is always finite, i.e. realistic. We use the Lagrangian method to solve for optimal $x$ and $y$. 
    \begin{align*}
         L(x,y,\lambda)= & x^2 +(x-y)^2(e^{1-p}-1)+(1-x)^2(e^p-1) \\
                         & +y+(1-y)^2 -\lambda(y-x) 
     \end{align*}
     \begin{align*}
        \frac{dL}{dx}= & 2x +2(x-y)(e^{1-p}-1)-2(1-x)(e^p-1)-\lambda = 0  \\
        \frac{dL}{dy}= & -2(x-y)(e^{1-p}-1)+1 +2(1-y) +\lambda = 0 \\ 
        \frac{dL}{d\lambda}= & y-x \leq 0 \qquad \lambda(y-x)=0 
    \end{align*}
    
    If $x \neq y$, the Kuhn-Tucker conditions are not solvable. If $x = y$ then we get that $x=y= \frac{e^p}{1+e^p}$. This tells us that at the social optimum, no players use the middle edge. For $p=0$ we get the social optimum of $x=y=\frac{1}{2}$ (the same solution as the case without traffic lights). Figure \ref{fig:scp} shows that the social cost functions of user equilibria and social optima are strictly increasing in $p$, so by choosing $p$ as low as possible we still reduce the costs of the best possible routing. 

      \begin{figure}[t]
          \centering
          \includegraphics[width=.49\textwidth, scale=0.7]{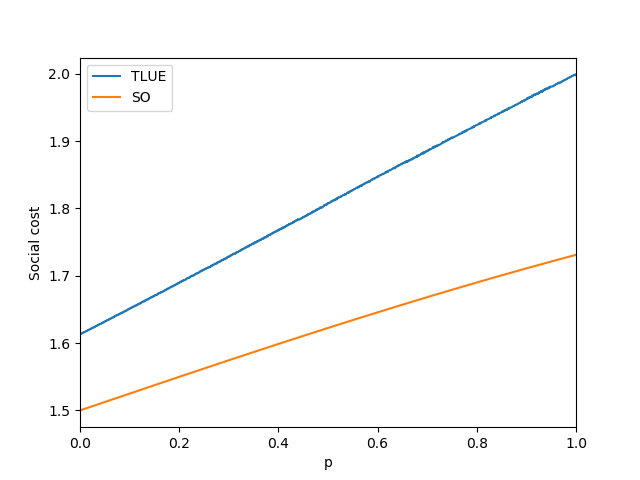}
          \caption{Social cost of TLUE and social optimum on the Wheastone network with a traffic light. The equilibrium costs are strictly less than the same network without a traffic light due to the altered routing behaviours. }
          \label{fig:scp}
     \end{figure}
     
     The TLUE never reaches the social optimum cost since the waiting time function is finite, so there are always some players using the middle edge. This property does not conincide  with the case when $p=0$, so again we see that there needs to be implemented some bounds on $p$ for realistic behaviours of the game: $0<p^-\leq p\leq p^+<1$.

     The UE cost of Braess' example where there exist no traffic lights (see Figure \ref{fig:Braess}) is 2. The traffic light is able to reduce the effects of selfish routing for all $p<1$. By including a traffic light, the cost of the TLUE is up to 24\% lower than without. Even with the additional constraint on $p$ to adhere to its bounds  $1-p^+<p<p^+$, the social cost is reduced by 22\%. This suggests that by positioning traffic lights at the end of edges which are not series-parallel, the costs of selfish routing are reduced. The price of anarchy belongs to $[1.07,1.15]$ for this game whereas once the traffic light and waiting time cost functions are removed the price of anarchy is, as well known (see e.g., \cite{rough}), strictly greater: $\frac{4}{3}$.

\section{Conclusions}\label{sec:conclusions}
    It is possible to model traffic lights in a routing game using expected waiting time functions where the equilibria of such games exist and are essentially unique. The proportion of red time at a traffic light influences routing choices. The waiting time functions behave non-monotonically past an upper bound value of the red light time in a cycle. It is therefore important to include a maximum red time length in intelligent traffic light models to minimise any unpredictable behaviour.
    
    By just altering the traffic light cycles, people's routing behaviour can change and alter the social cost of the equilibria significantly. Traffic lights can act as a central decision-maker to remove the effects of noncooperative selfish routing. By allowing for infinite waiting time functions, we can theoretically use traffic lights to make any network immune to Braess' paradox. Implementation of this is not compatible with the upper bounds for the proportion of red time in a traffic light cycle.

\bibliographystyle{named}
\bibliography{references}

\end{document}